\documentclass[10pt, a4paper,reqno]{amsart}

\usepackage{color} \definecolor{bleu_sombre}{rgb}{0,0,0.6}  \definecolor{rouge_sombre}{rgb}{0.8,0,0}\definecolor{vert_sombre}{rgb}{0,0.6,0}

\usepackage[plainpages=false,colorlinks,linkcolor=bleu_sombre,citecolor=rouge_sombre,urlcolor=vert_sombre,breaklinks]{hyperref}

\usepackage[T1]{fontenc}	
\usepackage[latin1]{inputenc}	
\usepackage[cm]{aeguill}	
\usepackage[english]{babel}

\usepackage{amsmath,amssymb,amsthm,amsfonts}
\usepackage{url,color}
\usepackage{nicefrac,enumerate,stmaryrd,dsfont}

\theoremstyle{plain}
\newtheorem{theorem}{{Theorem}}[section] 
\newtheorem*{theorem*}{{Theorem}}
\newtheorem{proposition}[theorem]{Proposition}
\newtheorem*{proposition*}{Proposition}

\newtheorem*{corollary*}{Corollary}
\newtheorem{lemma}[theorem]{Lemma}
\newtheorem*{lemma*}{Lemma}

\theoremstyle{definition}

\newtheorem*{definition*}{Définition}

\theoremstyle{remark}
\newtheorem*{remarque*}{Remarque}

\newtheorem*{exemple*}{Exemple}

\newtheorem*{exemples*}{Exemples}

\newcommand{\commm}[1]{{}}
\newcommand{\commperenne}[1]

\makeatletter

\@addtoreset{equation}{section}  
\makeatother

\renewcommand{\leq}{\leqslant}	\renewcommand{\geq}{\geqslant}
\renewcommand{\bar}[1]{\overline{#1}}
\renewcommand\over[2]{{\,\buildrel #1\over#2\,}}

\newcommand{\inv}{^{-1}}

\newcommand {\limt}[2]{\xrightarrow[#1 \to #2]{}}

\newcommand{\abs}[1]{\left\vert #1\right\vert}        
\newcommand{\nr}[1]{\left\Vert #1\right\Vert}         
\newcommand{\innp}[2]{\left< #1 , #2 \right>}         

\newcommand{\Op}{{\mathop{\rm{Op}}}_h}
\newcommand{\Opw}{{\mathop{\rm{Op}}}_h^w}

\newcommand{\pppg}[1] {\left< #1 \right>} 	
\newcommand{\symbor}{C_b^\infty}		
\newcommand{\symb} {\Sc}		

\newcommand{\bigo}[2]{\mathop{O}\limits_{#1 \to #2}}
\newcommand{\littleo}[2]{\mathop{o}\limits_{#1 \to #2}}

\newcommand{\singl}[1]{\left\{ #1 \right\}}		
\newcommand{\Ii}[2]{\llbracket #1,#2 \rrbracket}	
\newcommand{\R}{\mathbb{R}}		\newcommand{\C}{\mathbb{C}}
\newcommand{\N}{\mathbb{N}}

\newcommand{\st}{\,:\,}					

\newcommand{\seq}[2]{\left({#1}_{#2}\right)_{#2 \in\N}} 

\renewcommand{\Re}{\mathop{\rm{Re}}\nolimits}        
\renewcommand{\Im}{\mathop{\rm{Im}}\nolimits}        

\DeclareMathOperator{\Tr}{Tr}                        

\DeclareMathOperator{\supp}{supp}                    

\newcommand{\impl}{\Longrightarrow}                  

\renewcommand{\a}{\alpha}\renewcommand{\b}{\beta}\newcommand{\G}{\Gamma}\renewcommand{\d}{\delta}\newcommand{\D}{\Delta}\newcommand{\e}{\varepsilon}\newcommand{\z}{\zeta} \renewcommand{\th}{\theta}\newcommand{\Th}{\Theta}\renewcommand{\k}{\kappa}\renewcommand{\l}{\lambda}\newcommand{\m}{\mu}\newcommand{\n}{\nu}\newcommand{\x}{\xi}\newcommand{\s}{\sigma}\renewcommand{\t}{\tau}\newcommand{\vf}{\phi}\newcommand{\h}{\chi}\renewcommand{\o}{\omega}\renewcommand{\O}{\Omega}

\newcommand{\Lc}{{\mathcal L}}\newcommand{\Rc}{{\mathcal R}}\newcommand{\Sc}{{\mathcal S}}\newcommand{\Uc}{{\mathcal U}}\newcommand{\Vc}{{\mathcal V}}\newcommand{\Xc}{{\mathcal X}}

\newcommand{\hh}{H_h}   
   
\newcommand{\uh}{U_h}    \newcommand{\uhe}{U_h^E}

\newcommand{\zoneS}{\mathcal Z}

\newcommand{\negg}{{N_E \G}}
\newcommand{\nergg}{{N_E^R \G}}
\newcommand{\sgamma}{{\s_\G}}
\newcommand{\snegg}{{\s_{\negg}}}

\newcommand{\Opwm}{{\mathop{\rm{Op}}}_{h_m}^w}

\newcommand{\sigman} {\s_1}

\DeclareMathOperator{\II}{II}
\newcommand{\uhtr} {u_h^{T,R}}
\newcommand{\uhor} {u_h^{0,R}}
\newcommand{\uhr} {\tilde u_h^{R}}
\newcommand{\SR}{S^R_h}
\newcommand{\divg}{\mathop{\rm{div}}\nolimits}

\newcounter{stepproof}
\newcommand{\stepp}{\stepcounter{stepproof} \noindent {\bf $\bullet$}\quad }

\newcommand{\vfield}{\mathcal X}

\begin{document}

\title[Helmholtz equation with unbounded source]{Semiclassical measure for the solution of the Helmholtz equation with an unbounded source}
\author{Julien Royer}


\begin{abstract}
We study the high frequency limit for the dissipative Helmholtz equation when the source term concentrates on a submanifold of $\R^n$. We prove that the solution has a unique semi-classical measure, which is precisely described in terms of the classical properties of the problem. This result is already known when the micro-support of the source is bounded, we now consider the general case.
\end{abstract}

\maketitle

\vspace{2cm}

\section{Statement of the result}

We consider on $\R^n$ the Helmholtz equation
\begin{equation} \label{helmholtz}
(\hh-E_h) u_h = S_h, \quad \text{where } \hh = -h^2 \D + V_1(x) -ih V_2(x).
\end{equation}
Here $V_1$ and $V_2$ are smooth and real-valued potentials which go to 0 at infinity. Thus for any $h \in ]0,1]$ the operator $\hh$ is a non-symmetric (unless $V_2 = 0$) Schrödinger operator with domain $H^2(\R^n)$. The energy $E_h$ will be chosen in such a way that for $h >0$ small enough, $\d > \frac 12$ and $S_h \in L^{2,\d}(\R^n)$ the equation \eqref{helmholtz} has a unique outgoing solution $u_h \in L^{2,-\d}(\R^n)$. Here we denote by $L^{2,\d}(\R^n)$ the weighted space $L^2 \big(\pppg x ^{2\d} dx \big)$, where $\pppg x = \big ( 1 + \abs x ^2 \big)^{\frac 12}$.\\

The source term $S_h$ we consider is a profile which concentrates on a submanifold $\G$ of dimension $d \in \Ii 0 {n-1}$ in $\R^n$, endowed with the Lebesgue measure $\sgamma$. Given an amplitude $A \in C^\infty(\G)$ and $S$ in the Schwartz space $\Sc(\R^n)$ we set, for $ h \in ]0,1]$ and $x \in \R^n$:
\begin{equation} \label{def-source}
S_h(x) = h^{\frac {1-n-d}2}  \int_{\G} A(z) S\left( \frac {x-z} h\right) \, d\sgamma(z)
\end{equation}
(this definition will make sense with the assumptions on $\G$ and $A$ given below).
Our purpose is to study the semiclassical measures for the family of corresponding solutions $(u_h)$ when the submanifold $\G$ is allowed to be unbounded.\\

This work comes after a number of contributions which deal with more and more general situations. The first paper about the subject is \cite{benamouckp02}, where $\G = \singl 0$ (see also \cite{castella05}). The result was generalized in \cite{castellapr02} to the case where $\G$ is an affine subspace of $\R^n$, under the assumption that the refraction index is constant ($V_1 = 0$). This restriction was overcome in \cite{wangz06}. In \cite{fouassier06} the source term concentrates on two points and in \cite{fouassier06} the refraction index is discontinuous along an hyperplane of $\R^n$. All these papers study the semiclassical measure of the solution using its Wigner transform.

The approach in \cite{bony09} is different. The semiclassical measures are defined with pseudo-differential calculus (see \eqref{def-mesure}) and the resolvent is replaced by the integral over positive times of the propagator (as in \cite{castella05}). We used this point of view in \cite{art-mesure} to deal with the case of a non-constant absorption index ($V_2 \neq 0$, $V_2 \geq 0$) and a general bounded submanifold $\G$. We also considered in \cite{art-nondiss} an absorption index $V_2$ which can take non-positive values. The purpose of this paper is now to allow a general unbounded submanifold $\G$.\\

Let us now state more precisely the assumptions. The potentials $V_1$ and $V_2$ are respectively of long and of short range: there exist constants $\rho > 0$ and $c_\a$ for $\a \in\N^n$ such that 
\begin{equation} \label{hyp-decV}
\forall \a \in \N^n, \forall x \in \R^n, \quad \abs {\partial^\a V_1(x)} \leq c_\a \pppg x^{-\rho -\abs \a} \quad \text{and} \quad \abs {\partial^\a V_2(x)} \leq c_\a \pppg x^{-1 -\rho -\abs \a}.
\end{equation}

Then we introduce the hamiltonian flow $\vf^t$ corresponding to the classical symbol $p :(x,\x) \mapsto \abs \x^2 + V_1(x)$ on the phase space $\R^{2n}$. For all $w \in \R^{2n}$, $t \mapsto \vf^t(w) = (X (t,w), \Xi(t,w)) \in \R^{2n}$ is the solution of the system
\begin{equation} \label{syst-ham} 
\left\{ \begin{array}{l} \partial _t X (t,w) = 2 \Xi (t,w) ,  \\
        \partial_ t \Xi (t,w) = -  \nabla V_1 (X(t,w)), \\
	\vf^0(w) = w.
        \end{array} \right.
\end{equation}
For $I \subset \R$ we set 
\begin{eqnarray*}
\O_b(I) &=& \singl{ w \in p \inv(I) \st \sup_{t \in \R} \abs{X(t,w)} < \infty    }
%
\end{eqnarray*}
We also denote by 
\[
 H_p q = \{ p,q \} = \nabla_\x p \cdot \nabla_x q - \nabla _x p \cdot \nabla _\x q 
\]
the Poisson bracket of $p$ with a symbol $q \in C^\infty(\R^{2n})$.\\

We now consider an energy $E_0 > 0$ such that
\begin{equation} \label{hyp-amfaible}
\forall w \in \O_b(\singl {E_0}), \exists T > 0 , \quad \int_0^T V_2(X(t,w)) \, dt > 0.
\end{equation}
Let $\d > \frac 12$. We know (see \cite{art-nondiss}) that under Assumption \eqref{hyp-amfaible} there exist an open neighborhood $J$ of $E_0$, $h_0 > 0$ and $c \geq 0$ such that for 
\[
z \in \C_{J,+} = \singl{z \in \C \st \Re z \in J, \Im z > 0}
\]
and $h \in ] 0,h_0]$ the operator $(\hh -z)$ has a bounded inverse on $L^2(\R^n)$ and 
\[
 \nr{\pppg x ^{-\d} (\hh-z)\inv \pppg x^{-\d}}_{\Lc(L^2(\R^n))} \leq \frac c h.
\]
Here $\Lc(L^2(\R^n))$ denotes the space of bounded operators on $L^2(\R^n)$. Moreover for any $\l \in J$ the limit
\[
 (\hh-(\l +i0))\inv = \lim _{\b \to 0^+} (\hh-(\l +i\b))\inv 
\]
exists in $\Lc(L^{2,\d}(\R^n),L^{2,-\d}(\R^n))$.\\

Now let us be more explicit about the source term $S_h$ we consider. We recall that $\G$ is a submanifold of dimension $d \in \Ii 0 {n-1}$ in $\R^n$, endowed with the Riemannian structure given by the restriction of the usual structure on $\R^n$, and the corresponding Lebesgue measure $\sgamma$. We assume that there exist $R _ 1 >0$ and $\s_1 \in ]0,1[$ such that
\begin{equation} \label{hyp-non-incoming}
N\G \cap \zoneS_-(R_1 , 0 , -\s_1) = \emptyset,
\end{equation}
where $N\G = \singl{(z,\x) \in \G \times \R^n \st \x \bot T_z \G}$ is the normal bundle of $\G$ and $\zoneS_-(R_1, 0 , \s_1)$ is an incoming region: for $R \geq 0$, $\n \geq 0$ and $\s \in [-1,1]$ we set
\[
\zoneS_\pm(R,\n,\s)  = \singl{ (x,\x)\in\R^{2n} \st \abs x \geq R,   \abs \x  \geq \n  \text{ and } \innp x \x \gtreqless \s \abs x \abs \x} .
\]
Note that Assumption \eqref{hyp-non-incoming} is satisfied for any bounded submanifold of $\R^n$. When $d = 0$, it actually implies that $\G$ is bounded, but this is not the case in higher dimension: this assumption holds for instance for any affine subspace of dimension $d \in \Ii 1 {n-1}$ in $\R^n$. Now that $\s_1$ is fixed, we can choose a smaller neighborhood $J$ of $E_0$ and assume that
\begin{equation} \label{hyp-J}
J \subset [E_1,E_2], \quad \text{were} \quad   E_1 > 0 \quad \text{and} \quad \left( \frac {1 + \s_1}2\right)^2 E_2 < E_1.
\end{equation}\\

We assume that 
\begin{equation} \label{hyp-propag}
 \forall z \in \G , \quad V_1(z) > E_0,
\end{equation}
and we define
\[
\negg = N\G \cap p\inv(\singl{E_0}).
\]
$\negg$ is a submanifold of dimension $(n-1)$ in $\R^{2n}$, endowed with the Riemannian structure defined as follows: for $(z,\x) \in \negg$ and $(Z,\Xi),(\tilde Z,\tilde \Xi) \in T_{(z,\x)}\negg \subset \R^{2n}$ we set
\[
 g_{(z,\x)} \big((Z,\Xi),(\tilde Z,\tilde \Xi) \big) = \innp{Z}{\tilde Z}_{\R^n} + \innp{ \Xi_\bot}{\tilde \Xi_\bot}_{\R^n},
\]
where $\Xi_\bot,\tilde \Xi_\bot$ are the orthogonal projections of $\Xi,\tilde \Xi \in \R^n$ on $(T_z\G \oplus \R \x)^\bot$ (see the discussion in \cite{art-mesure}).
We denote by $\snegg$ the canonical measure on $\negg$ given by $g$, and assume that
\begin{equation} \label{hyp-Phi}
 \snegg \left( \singl{ (z,\x) \in \negg \st \exists  t > 0 , \vf^t (z,\x) \in \negg} \right) = 0.
\end{equation}

We now introduce the amplitude $A \in C^\infty(\G)$. We assume that there exist $\d > \frac 12$ and $c \geq 0$ such that 
\begin{subequations} \label{hyp-decA}
\begin{equation} \label{hyp-decAb} 
\int_{\G}\pppg z ^{\d} \big( \abs {A(z)} + \nr{d_z A} + \abs {A(z)} \nr{\II_z}\big) \, d\sgamma (z) < +\infty .
\end{equation}
Moreover for all $r \in ]0,1]$ and $x \in \R^n$ we have
\begin{equation} \label{hyp-decAa} 
 \int_{B(x,r)\cap \G} \pppg z ^{\d} \big( \abs{A(z)} + \nr{d_z A}  + \abs {A(z)} \nr{\II_z}\big) \, d\sgamma (z) \leq c r^d.
\end{equation}
\end{subequations}
Here $B(x,r)$ is the ball of radius $r$ and centered at $x$, $d_z A : T_z\G \to \R$ is the differential of $A$ at point $z$ and $\II$ is the second fundamental form of the submanifold $\G$. For any $z \in \G$, $\II_z$ is a bilinear form from $T_z\G$ to $N_z\G$ (see Appendix \ref{sec-II}), and
\[
 \nr{\II_z} =\sup_{\nr{X}_{T_z\G} = \nr{Y}_{T_z\G} = 1} \nr{\II_z(X,Y)}_{N_z\G}.
\]
Note that all these estimates hold when $A \in C_0^\infty(\G)$. 
Here $A$ is allowed to have a non-compact support, but it still has to stay away from the bundary of $\G$:
\begin{equation} \label{hyp-Abord}
\forall \th \in C_0^\infty(\R^n), \quad z \mapsto A(z) \th(z) \in C_0^\infty(\G).
\end{equation}
Then it remains to consider $S \in \Sc(\R^n)$ and define the source term $S_h$ by \eqref{def-source}.\\ \commperenne{On ne définit pas $\G_0$ et $\negg_0$.}

Let $(E_h) _{h \in ]0,h_0]}$ be a family of energy in $\C_{J,+} \cup J$ such that 
\begin{equation} \label{hyp-Eh}
 E_h = E_0 + h \tilde E +\littleo h 0 (h)
\end{equation}
for some $\tilde E \in \bar {\C_+}$. Since for all $h \in ]0,h_0]$ the source term $S_h$ given by \eqref{def-source} belongs to $L^{2,\d}(\R^n)$ (see Proposition \ref{prop-norme-S}) we can define
\begin{equation} \label{def-uh}
u_h = (\hh-(E_h +i0))\inv S_h \in L^{2,-\d}(\R^n).
\end{equation}
Here $(\hh-(E_h +i0))\inv$ stands for $(\hh-E_h)\inv$ when $E_h \in \C_{J,+}$.\\

Our purpose is to study the semiclassical measures for this family $(u_h)_{h\in]0,h_0]}$. In other words, non-negative measures $\m$ on the phase space $\R^{2n} \simeq T^* \R^n$ such that 
\begin{equation} \label{def-mesure}
\forall q \in C_0^\infty(\R^{2n}), \quad \innp{\Opwm(q) u_{h_m}}{u_{h_m}} \limt m {+\infty} \int _{\R^{2n}} q \, d\m,
\end{equation}
for some sequence $\seq h m \in ]0,h_0]^\N$ such that $h_m \to 0$ (see \cite{gerard91}). Here $\Opw(q)$ denotes the Weyl quantization of the symbol $q$:
\[
 \Opw(q) u (x) = \frac 1{(2\pi h)^n} \int_{\R^n} \int_{\R^n} e^{\frac ih \innp{x-y} \x} q \Big( \frac {x+y} 2 , \x \Big) u(y) \, dy \, d\x.
\]
We will also use the standard quantization:
\[
 \Op(q) u (x) = \frac 1{(2\pi h)^n} \int_{\R^n} \int_{\R^n} e^{\frac ih \innp{x-y} \x} q (x,\x) u(y) \, dy \, d\x.
\]
We denote by $\symbor(\R^{2n})$ the set of smooth symbols whose derivatives are bounded. For $\d \in \R$, we also denote by $\Sc\big(\pppg x ^\d \big)$ the set of symbols $a \in C^\infty(\R^{2n})$ such that
\[
 \forall \a,\b\in\N^n, \exists c_{\a,\b} \geq 0, \forall (x,\x) \in \R^{2n} , \quad \abs{\partial_x^\a \partial_\x^\b a(x,\x)} \leq c _{\a,\b} \pppg x^{\d},
\]
and by $\symb_{\d}(\R^{2n})$ the set of symbols $a \in C^\infty(\R^{2n})$ such that
\[
 \forall \a,\b\in\N^n, \exists c_{\a,\b} \geq 0, \forall (x,\x) \in \R^{2n} , \quad \abs{\partial_x^\a \partial_\x^\b a(x,\x)} \leq c _{\a,\b} \pppg x^{\d - \abs \a}.
\]
We can similarly define the sets of symbols $\symb \big( \pppg \x^\d \big)$ for $\d \in \R$. We refer to \cite{robert, zworski} for more details about semiclassical analysis.\\

For $(z,\x) \in \negg$ we set 
\[\k(z,\x) = \pi (2\pi)^{d-n}  \abs{A(z)}^2 \abs \x \inv \big|\hat S (\x) \big| ^2,\]
where $\hat S$ is the Fourier transform of $S$. The theorem we want to prove is the following:

\begin{theorem} \label{th-mesure-nb}
Let $S_h$ and $E_h$ be given by \eqref{def-source} and \eqref{hyp-Eh}, and $u_h$ defined by \eqref{def-uh}. Assume that the assumptions \eqref{hyp-decV}, \eqref{hyp-amfaible}, \eqref{hyp-non-incoming}, \eqref{hyp-J}, \eqref{hyp-propag}, \eqref{hyp-Phi}, \eqref{hyp-decA} and \eqref{hyp-Abord} hold.
\begin{enumerate}[(i)]
\item Then there exists a non-negative Radon measure $\m$ on $\R^{2n}$ such that for all $q \in C_0^\infty(\R^{2n})$ we have
\begin{equation}  \label{lim-opquu}
\innp{\Opw(q) u_h}{u_h} \limt h 0  \int_{\R^{2n}} q \, d\m.
\end{equation}
\item This measure is characterized by the following three properties:
\begin{enumerate} [a.]
\item $\m$ is supported in $p\inv(\singl{E_0})$.
\item For all $\s \in ]\sigman,1[$ there exists $R \geq 0$ such that $\m$ is zero in the incoming region $\zoneS_-(R,0,-\s)$.
\item $\m$ satisfies the following Liouville equation:
\begin{equation*} 
(H_p + 2 \Im \tilde E + 2 V_2) \m = \k  \, \snegg.
\end{equation*}
This means that for all $q \in C_0^\infty(\R^{2n})$ we have
\[
\int_{\R^{2n}} (-H_p + 2 \Im \tilde E + 2 V_2) q \, d\m = \int_{\negg} q(z,\x) \k(z,\x) \, d\snegg(z,\x).
\]
\end{enumerate}
\item These three properties imply that for any $q \in C_0^\infty(\R^{2n})$ we have
\begin{equation} \label{expr-mu}
\int_{\R^{2n}} q \, d\mu = \int_{\R_+} \int_{\negg}  \k(z,\x)  q(\vf^t(z,\x)) e^{-2t\Im \tilde E - 2\int_0^t V_2(X(s,z,\x))\,ds} \,d\snegg(z,\x) \,dt.
\end{equation}
\end{enumerate}
\end{theorem}

This result is known when $A$ is compactly supported on $\G$ (see \cite{art-mesure,art-nondiss}). The idea for the proof is to write the resolvent as the integral of the propagator over positive times, and to approximate $u_h$ by a partial solution which only takes into account finite times:
\[
 u_h^T = \frac ih \int_0^\infty \h_T(t) e^{-\frac {it}h (\hh-E_h) } S_h \,dt.
\]
Here $\h_T(t) = \h(t-T)$, where $\h \in C^\infty(\R, [0,1])$ is equal to 1 in a neighborhood of $]-\infty, 0]$ and supported in $]-\infty,\t_0]$ for some well-chosen $\t_0 \in ]0,1]$. Note that for all $h \in ]0,1]$ the semi-group $t \mapsto e^{-\frac {it}h \hh}$ is well-defined for all $t \geq 0$. However this is not a contractions semi-group since $V_2$ is not assumed to be non-negative (see for instance Corollary 3.6 in \cite{engel})\\

The idea will be the same to deal with the case of an amplitude $A$ whose support is not bounded. Let $\Th \in C_0^\infty(\R^n,[0,1])$ be equal to 1 on $B(0,1)$. For any $R > 0$ we set $A_R : z \in \G \mapsto A(z) \Th(z/R)$ and 
\begin{equation} \label{def-shr}
\SR(x) = h^{\frac {1-n-d} 2} \int_{z\in\G} A_R(z) S\left( \frac {x-z}h \right)\, d\sgamma(z).
\end{equation}

Given any $R > 0$, the proof of \cite{art-mesure,art-nondiss} applies for the source term $\SR$. Since the choice of $\h$ mentionned above depends on the support of $A_R$, we denote it by $\h_{0,R}$. Moreover $\h_{0,R}$ can be chosen non-increasing. Then for any $T \geq 0$ we set $\h_{T,R} : t \mapsto \h_{0,R} (t-T)$ and for any $h \in ]0,h_0]$:
\[
\uhtr = \frac ih \int_0^\infty \h_{T,R}(t) e^{-\frac {it}h (\hh-E_h)} \SR \, dt.
\]
The key point is to prove that in some suitable sense $\uhtr$ is a good approximation of $u_h$ for large $T$ and $R$, and $h>0$ small enough.\\

Let $R > 0$. For $h \in ]0,h_0]$ we set
\[
\uhr = (\hh-(E_h+i0))\inv \SR \in L^{2,-\d}(\R^n).
\]
Since $A_R$ is compactly supported on $\G$, we know that Theorem \ref{th-mesure-nb} holds for $\tilde u_h^R$. In particular there exists a non-negative Radon measure $\tilde \m_R$ on $\R^{2n}$ such that 
\[
\forall q \in C_0^\infty(\R^{2n}), \quad \innp{\Opw(q) \uhr}{\uhr} \limt h 0 \int_{\R^{2n}} q \, d \tilde \m_R.
\]
Moreover, according to \eqref{expr-mu}, $\tilde \m_R$ is supported on the classical trajectories coming from $\nergg =\singl{(z,\x) \in \negg \st z \in \supp A_R}$. Let $K$ be a compact subset of $\R^{2n}$. Assumption \eqref{hyp-non-incoming} ensures that for $R_K > 0$ large enough and $R \geq R_K$, the trajectories coming form $\nergg \setminus N_E^{R_K}\G$ do not meet $K$ (see Proposition \ref{prop-non-entr}), and hence $\tilde \m_{R} = \tilde \m_{R_K}$ on $K$. This is the idea we are going to use to prove existence of the semiclassical measure $\m$. And as expected, $\m$ will coincide with $\tilde \m_{R_K}$ on $K$.\\

The plan of this paper is the following. In section \ref{sec-terme-source} we give some estimates for the source term $S_h$, and in Section \ref{sec-approx} we show that $\uhtr$ is a good approximation of $u_h$ in order to prove Theorem \ref{th-mesure-nb}. In Appendix \ref{sec-II} we recall some basic facts about differential geometry, and in particular the second fundamental form which appears when integrating by parts on $\G$.

\section{Estimates of the source term} \label{sec-terme-source}

In this section we prove that $S_h$ and $\SR$ for $R > 0$ are (uniformly) of size $O(\sqrt h)$ in $L^{2,\d}(\R^n)$, where $\d > \frac 12$ is given by \eqref{hyp-decA}. Then we use Assumption \eqref{hyp-non-incoming} to prove that if $\o_- \in \symb_0(\R^{2n})$ is supported in $\zoneS_-(R,0,-\s)$ for some $R > R_1$ and $\s \in ]\s_1,1[$, then $\Opw(\o_-) S_h = O(h^{\frac 32})$ in $L^{2,\d}(\R^n)$. Since we even have an estimate of size $O(h^\infty)$ when $\o_-$ is compactly supported, this proves in particular that $S_h$ is microlocally supported outside an incoming region.

\begin{lemma} \label{estim-sqrt}
Consider $B \in C^\infty(\G)$, a family $(f_h^z)_{z \in\G , h\in]0,1]}$ of functions in $L^{2, 2 + \d + d/2 }(\R^n)$, and assume that for some $C_1 \geq 0$ we have
\begin{equation} \label{hyp-decBb}
\forall h \in ]0,1], \quad \int_{\G}  \pppg z ^\d \abs{B(z)} \left( 1 + \nr {f_h^z}^2_{L^{2, 2 + \d + d/2 }(\R^n)} \right) \,  d\sgamma(z) < C_1
\end{equation}
and 
\begin{equation} \label{hyp-decBa}
\forall x \in \R^n , \forall r \in ]0,1], \quad  \int_{B(x,r) \cap \G}  \pppg z ^\d \abs{B(z)} \, d\sgamma(z) \leq C_1 \, r^d.
\end{equation}
For $h \in ]0,1]$ we consider
\[
\tilde S_h : x \mapsto  h^{\frac {1-n-d}2} \int_{\G} B(z) f_h^z\left( \frac {x-z}h\right) \,d\sgamma (z).
\]
Then $\tilde S _h(x)$ is well-defined for all $h \in ]0,1]$ and $x \in\R^n$, and there exists a constant $C \geq 0$ which only depends on $C_1$ and such that 
\[
\forall h \in ]0,1], \quad \nr{\tilde S_h}_{L^{2,\d}(\R^n)} \leq C \sqrt h.
\]
\end{lemma}

The idea of the proof is the same as for a compactly supported amplitude but we now have to be careful with the decay at infinity for functions on $\G$.

\begin{proof}
We first remark that Assumption \eqref{hyp-decBa} holds for all $r > 0$ since for $x \in \R^n$ and $r \geq 1$ Assumption \eqref{hyp-decBb} gives
\begin{equation} \label{hyp-decB-r}
\int_{B(x,r)\cap \G} \pppg z^\d \abs{B(z)} \, d\sgamma(z) \leq C_1 \leq r^d C_1.
\end{equation}
Let $h\in ]0,1]$ and $x \in \R^n$. According to Cauchy-Schwarz inequality and \eqref{hyp-decB-r} we can write
\begin{eqnarray*}
\lefteqn{\left(\int_{\G} \abs{B(z)}\abs{ f_h^z \left(\frac {x-z}h\right)} \, d\sgamma(z)\right)^2  = \left(\sum_{m\in\N} \int_{mh\leq \abs{x-z}< (m+1)h} \abs{B(z)} \abs{ f_h^z \left( \frac {x-z}h\right)} \, d\sgamma(z)\right)^2}\\
&\hspace{2cm} & \leq c   \sum_{m\in\N} \pppg m^2  \left( \int_{mh\leq \abs{x-z}< (m+1)h} \abs{B(z)} \abs{ f_h^z \left( \frac {x-z}h\right)} \, d\sgamma(z) \right)^2\\
&& \leq c \, C_1 \,  h^d \sum_{m\in\N} \pppg m ^{2+d} \int_{mh\leq \abs{x-z}< (m+1)h} \pppg z ^{-\d} \abs{B(z)} \abs{ f_h^z \left( \frac {x-z}h\right)}^2 \, d\sgamma(z),
\end{eqnarray*}
where $c \geq 0$ stands for different universal constants. 
Now using \eqref{hyp-decBb} we obtain
\begin{eqnarray*}
\lefteqn{\nr{S_h}^2_{L^{2,\d}(\R^n)}}\\
&& \leq c\, C_1 \, h^{1-n} \int_{\R^n}\pppg x ^{2\d}\sum_{m\in\N}  \pppg m^{2+d} \int_{mh\leq \abs{x-z}< (m+1)h}  \pppg z^{-\d} \abs{B(z)} \abs{ f_h^z \left( \frac {x-z}h\right)}^2 \, d\sgamma(z)\, dx\\
&& \leq c\, C_1 \, h \sum_{m\in\N} \pppg m^{2+d} \int_{\G} \int_{m\leq \abs{y}< m+1} \pppg {z+hy} ^{2\d} \pppg z ^{-\d} \abs{B(z)} \abs{ f_h^z (y)}^2 \, dy \, d\sgamma(z)\\
&& \leq c\,  C_1 \, h \sum_{m\in\N} \pppg m ^{-2} \int_{\G} \pppg {z} ^{\d}\abs{B(z)}  \int_{m\leq \abs{y}<m+1} \pppg {y} ^{4+d+2\d}  \abs{ f_h^z (y)}^2 \, dy\, d\sgamma(z) \\
&& \leq c \, C_1^2 \, h.
\end{eqnarray*}
\end{proof}

Applied with $f_h^z = S$ and $B = A$ or $A_R$ for $R > 0$ this proposition gives:

\begin{proposition} \label{prop-norme-S}
There exists a constant $C \geq 0$ such that 
\begin{equation*}
\forall h \in ]0,1], \forall R > 0, \quad \nr{S_h}_{L^{2,\d} (\R^n)}+ \nr{\SR}_{L^{2,\d} (\R^n)}  \leq C \sqrt h.
\end{equation*}
\end{proposition}

For $z \in \G$ and $\x \in \R^n$ we denote by $\x_z^T$ the orthogonal projection of $\x$ on $T_z\G$, and $\x_z^N = \x - \x_z^T$.

\begin{proposition} \label{prop-microloc++}
Let $q  \in \symbor (\R^{2n})$ and assume that for some $\e > 0$ we have
\[
\forall (x,\x) \in \supp q, \forall z \in \supp A, \quad \big( \abs{x - z} \geq \e \quad \text{or} \quad \abs{\x_z^T} \geq \e  \big) .
\]
Then we have
\[
\nr{\Op(q) S_h}_{L^{2,\d}(\R^n)} = \bigo h 0 \big(h^{\frac 32}\big).
\]
\end{proposition}

\begin{proof} 
\stepp
We have
\begin{align*}
\Op (q) S_h(x)
 = \frac {h^{\frac {1-n-d}2}}  {(2\pi h)^n} \int _{\R^n} \int_{\R^n} \int _\G e^{\frac ih \innp{x-y} \x} q( x,\x ) A(z) S\left( \frac {y-z} h\right) \, d\sgamma(z) \, dy \, d\x.
\end{align*}
We recall that this is defined in the sense of an oscillatory integral. After a finite number of partial integrations with the operator $\frac {1 + ih \x \cdot \nabla_y}{1 +  \abs \x^2}$ we can assume that $q \in \symb \big(\pppg \x^{-(n+1)}\big)$. Then we can use Fubini's Theorem and make the change of variables $y = z + h v$ for any fixed $z$. 
%
%
Let $\h_1 \in C_0^\infty(\R^n,[0,1])$ be supported in $B(0,\e)$ and equal to 1 on a neighborhood of 0. We set $\h_2 = 1 - \h_1$. For $x \in \R^n$ and $h \in ]0,1]$ we can write
\[
\Op (q) S_h(x) = I_1(x,h) + I_2(x,h),
\]
where for $j \in \{1,2\}$:
\[
I_j(x,h)  =  \frac {h^{\frac {1-n-d}2}}  {(2\pi)^n}  \int _\G A(z)  \h_j(x-z)  \int _{\R^n} \int_{\R^n} e^{\frac ih \innp{x-z} \x} e^{-i \innp v \x}  q( x,\x ) S(v) \, dv \, d\x\, d\sgamma(z).\\
\]

\stepp
Let $N \in\N$ and
\[
B(x,z,h) = \frac   {\h_2 (x-z)} {(2\pi)^{n}}  \int _{\R^n} \int_{\R^n} e^{\frac ih \innp{x-z} \x}S(v)  L^N \left(e^{-i \innp v \x}  q( x,\x )\right) \, dv \, d\x,
\]
where for $f \in C^\infty(\R^n \times \G \times \R^n \times ]0,1])$ we have set
\[
 L v : (x,z,\x,h) \mapsto  -ih \divg _\x \left( \frac {(x-z) f(x,z,\x,h)}{\abs{x-z}^2 } \right).
\]
Then there exists a constant $c$ such that
\[
\forall h \in ]0,1], \forall x \in \R^n , \forall z \in \G, \quad \abs{B(x,z,h)}  \leq c  {h^N}\pppg {x-z}^{-N},
\]
and hence
\begin{align*}
\abs{I_2(x,h)}^2
& = h^{1-n-d}  \abs{ \int _{\G} A(z)  B(x,z,h)   \, d\sgamma(z) }^2\\
& \leq c  \, h^{2N + 1 - n - d} \left(\int_{\G} \pppg z^\d  \abs{A(z)}  \, d\sgamma(z) \right) \left(  \int_{\G} \pppg z ^ {-\d} \abs{A(z)}  \pppg{x-z}^{-2N } \, d\s(z)\right),
\end{align*}
where $c$ depends neither on $x \in\R^n $ nor on $h \in ]0,1]$. We obtain
\begin{align*}
\nr{I_2(h)}^2_{L^{2,\d}(\R^n)}
& \leq c\, h^{2N + 1 - n - d}\int_{\G} \int_{\R^n} \pppg z ^ {-\d} \abs{A(z)} \pppg x^{2\d} \pppg{x-z}^{-2N}  \, dx \, d\s(z) \\
& \leq c\, h^{2N + 1 - n - d}\int_{\G} \int_{\R^n} \pppg z ^ {\d} \abs{A(z)} \pppg {x-z}^{2\d} \pppg{x-z}^{-2N}  \, dx \, d\s(z) 
\end{align*}
and finally, if $N$ was chosen large enough:
\[
\nr{I_2(h)}^2_{L^{2,\d} (\R^n)} = \bigo h 0 \big(h^3\big) .
\]

\stepp
We now turn to $I_1$. Let $x,\x \in \R^n$. The function $z \mapsto \x_z^T \in T_z \G$ defines a vector field on $\G$, which we denote by $\x^T$, and the norm of $\x_z^T$ in $T_z\G$ is the same as in $\R^n$. By assumption, if $A(z) q(x,\x) \h_1(x-z) \neq 0$ then $\abs {\x_z^T} \geq \e$. And we remark that when $\x_z^T \neq 0$ we have
\[
e^{\frac ih \innp{x-z}\x}  =  ih{\abs{\x_z^T}^{-2}} \x_z^T \cdot  e^{\frac ih \innp{x-z}\x},
\]
where $\x_z^T \cdot f(z)$ is the derivative of $f \in C^\infty(\G)$ at point $z$ and in the direction of $\x_z^T$.
For any $z \in \supp A \cap \bar {B(x,\e)}$ (which is compact according to Assumption \eqref{hyp-Abord}) there exists an open neighborhood $\Vc_z$ of $z$ in $\G$ which is orientable, and we can find $z_1,\dots,z_K \in \supp A \cap \bar {B(x,\e)}$ such that $\supp A \cap \bar {B(x,\e)} \subset \bigcup_{k=1}^K \Vc_{z_k}$. We consider $\z_1,\dots, \z_K \in C_0^\infty(\G,[0,1])$ such that $\sum_{ k=1}^K \z_k = 1$ in a neighborhood of $\supp A \cap \bar {B(x,\e)}$ in $\G$ and $\z_k$ is supported in $\Vc_k$ for all $k \in \Ii 1 K$.
Let $k\in\Ii 1 K$. According to Green's Theorem \ref{th-green} we have
\[
\int _\G \divg \left( e^{\frac ih \innp{x-z}\x} \h_1(x-z) (A\z_k)(z) \abs{\x_z^T}^{-2} \x_z^T\right) \, d\sgamma(z) = 0
\]
and hence
\begin{eqnarray*}
\lefteqn{\int_{\G} e ^ {\frac ih \innp{x-z}\x}  (A \z_k) (z) \h_1 (x-z) \, d\sgamma(z) }\\
&& = - ih\int_{\G} \divg \x^T(z)  {\abs{\x_z^T}^{-2}} e^{\frac ih \innp{x-z}\x}  (A \z_k) (z)  \h_1( x-z) \, d\sgamma(z)\\
&& \quad - ih \int_{\G} e^{\frac ih \innp{x-z}\x} {\x_z^T  }\cdot \big((A\z_k)(z)  \h_1( x-z)\big)  {\abs{\x_z^T}^{-2}} \, d\sgamma(z)\\
&& \quad  + ih \int_{\G} e^{\frac ih \innp{x-z}\x} (A\z_k)(z)  \h_1( x-z) {\abs{\x_z^T}^{-4}} \x_z^T \cdot  \abs{\x_z^T}^2  \, d\sgamma(z).
\end{eqnarray*}
Taking the sum over $k \in \Ii 1 K$ gives
\[
I_{1}(h) = -ih \big(I_{1,1}(h) + I_{1,2}(h) + I_{1,3}(h) \big)
\]
where, for instance,
\begin{eqnarray*}
\lefteqn{I_{1,1}(x,h)}\\
&& =  \frac {h^{\frac {1-n-d}2}}  {(2\pi)^n}  \int _\G A(z)   \h_1(x-z)  \int _{\R^n} \int_{\R^n} \divg \x^T(z)  {\abs{\x_z^T}^{-2}} e^{\frac ih \innp{x-z} \x} e^{-i \innp v \x} q(x,\x) S(v) \, dv \, d\x\, d\sgamma(z)\\
&& =  \frac {h^{\frac {1-n-d}2}}  {(2\pi)^n}  \int _\G A(z)\pppg {\nr{\II_z}}   \h_1(x-z)  \int _{\R^n} \int_{\R^n}  e^{\frac ih \innp{x-z} \x} e^{-i \innp v \x} q_{1,z}(x,\x) S(v) \, dv \, d\x\, d\sgamma(z)\\
&& = h^{\frac {1-n-d}2}    \int _\G A(z) \pppg {\nr{\II_z}} \h_1(x-z) \big(\Op(q_{1,z}) S \big) \left(\frac {x-z}h \right)\, d\sgamma(z) .
\end{eqnarray*}
Here we have set
\[
q_{1,z}(x,\x) = \pppg {\nr{\II_z}} \inv  \divg \x^T(z)  {\abs{\x_z^T}^{-2}}q(x,\x).
\]
%


\stepp
We recall that the Levi-Civita connection $\nabla ^\G$ on the submanifold $\G$ at point $z$ is given by the orthogonal projection on $T_z\G$ of the usual differential on $\R^n$ (see Appendix \ref{sec-II}).
Let $z \in \G$. Let $Y$ be a vector field on a neighborhood of $z$ in $\G$ and let $\bar Y$ be an extension of $Y$ on a neighborhood of $z$ in $\R^n$.
Using \eqref{der-innp} we see that on a neighborhood of $z$ in $\G$ we have
\begin{align*}
\innp{\nabla_{Y} \x^T} {Y}_ {T \G}
& = Y \cdot \innp{ \x^T}{Y}_ {T\G} - \innp{\x^T} {\nabla_{Y} Y}_ {T\G} = \bar Y \cdot \innp{ \x}{\bar Y}_ {\R^n}- \innp{\x^T} {\nabla_{Y} Y}_ {T\G}\\
& = \innp{\x} {\nabla_{Y}Y + \II(Y,Y)}_{\R^n}- \innp{\x^T} {\nabla_{Y} Y}_ {T\G}\\
& = \innp{\x^N} {\II(Y,Y)}_{\R^n}.
\end{align*} 
Now if $Y_1,\dots, Y_d$ are vector fields such that $(Y_1(z),\dots, Y_d(z))$ is an orthonormal basis of $T_z\G$ we obtain
\begin{align*}
\divg \x^T (z) = \Tr \left( Y \mapsto \nabla_Y \x^T (z)\right) = \sum _{j=1}^d \innp{\x_z^N} {  \II_z (Y_j(z),Y_j(z))}_ {\R^n} ,
\end{align*}
and hence
\[
 \abs{\divg \x^T(z)} \leq d \abs{\x_z^N} \nr{\II_z}.
\]
We can apply Lemma \ref{estim-sqrt} with $B = A\pppg {\nr{\II_z}}$ and $f_h^z(x) = \h_1(h x) (\Opw(q_{1,z}) S)(x)$. Indeed, since $q$ is assumed to be in $\symb(\pppg \x\inv)$, $q_{1,z}$ is in $\symbor(\R^{2n})$ uniformly in $z \in \G$, and hence for all $ k\in\N$ the operator $\Opw(q_{1,z})$ belongs to $\Lc(L^{2,k}(\R^n))$ and the norm is uniform in $z$ (see \cite{wang88}). This proves that
\[
 \nr{I_{1,1}(h)}_{L^{2,\d}(\R^n)} = \bigo h 0 \big( \sqrt h \big).
\]

\stepp
We have
\[
 \x_z^T \cdot \big( A(z) \h_1(x-z) \big)=   \h_1(x-z) d_zA(\x_z^T) +  A(z)  \, \x_z^T \cdot \h_1(x-z).
\]
%
%
%
We set 
\[
\begin{cases}
 B_{2,1}(z) = \nr {d_zA} \\
q_{z,2,1}(x,\x) =  \h_1(hx)d_zA(\x_z^T)\nr {d_zA}\inv  q(x,\x)
\end{cases}
\quad \text {and} \quad 
\begin{cases}
 B_{2,2}(z) = A(z)\\
q_{z,2,2}(x,\x) = q(x,\x) d(\h_1)_{hx}(\x_z^T)
\end{cases}
\]
($d_zA(\x_z^T)\nr {d_zA}\inv$ can be replaced by $\abs{\x}$ when $\nr {d_zA}=0$). As above, applying Lemma \ref{estim-sqrt} with $B_{2,j}$ and $f_{h,2,j}^z =\Opw(q_{z,2,j}) S$ for $j \in \{1,2\}$ gives then
\[
 \nr{I_{1,2}(h)}_{L^{2,\d}(\R^n)} = \bigo h 0 \big( \sqrt h \big).
\]

\stepp
We now deal with $I_{1,3}(h)$. For any vector field $Y$ on $\G$ we have
\begin{align*}
Y \cdot \innp{\x^T}{\x^T}_{T\G} 
& = Y \cdot \innp{\x}{\x^T}_{\R^n}\\
&  =  \innp{\x}{ \nabla_Y \x^T + \II(Y,\x^T)} = \innp{ \x^T}{\nabla_Y \x^T}_{T \G} + \innp{\x^N}{\II(Y,\x^T)}_{\R^n}\\
& = \frac 12 Y \cdot \innp{ \x^T}{\x^T}_{T \G} + \innp{\x^N}{\II(Y,\x^T)}_{\R^n} ,
\end{align*}
and in particular:
\[
\x^T \cdot  \innp{\x^T}{\x^T}_{T\G}  = 2  \innp{\x^N}{\II \big(\x^T,\x^T\big)}_{\R^n}.
\]
Thus we can estimate $I_{1,3}(h)$ as $I_{1,1}(h)$, and this concludes the proof.
\end{proof}

We now check that according to \eqref{hyp-non-incoming} the assumptions of Proposition \ref{prop-microloc++} are satisfied for a symbol $q$ supported in an incoming region:

\begin{proposition} \label{prop-microloc-inf}
Let $\s_2 \in ]\s_1,1]$, $R_2 > R_1$ and $\n_0 > 0$. Then there exists $\e > 0$ such that for all $(x,\x) \in \zoneS_-(R_2,\nu_0,-\s_2)$ and $z \in \G$ we have
\[
\abs{x - z} \geq \e \quad \text{or} \quad \abs{\x_z^T} \geq \e.
\]
\end{proposition}

\begin{proof}
Let $(x,\x) \in \zoneS_-(R_2,\n_0,-\s_2)$, $z \in \G$, and assume that 
\[
 \abs {x-z} \leq \min \left( R_2 - R_1, \frac {R_1(\s_2-\s_1)}4 \right).
\]
In particular $\abs z \geq R_1$ and hence, according to \eqref{hyp-non-incoming}, we have
\begin{align*}
\abs{z} \abs{\x_z^T}
&  \geq - \innp z {\x_z^T} = \innp{z}{\x_z^N - \x} = \innp z {\x_z^N}- \innp x \x + \innp{x-z}{\x} \\
& \geq - \s_1 \abs z \abs {\x_z^N}  + \s_2 \abs x \abs \x  - \abs{x-z} \abs \x
\geq \left( - \s_1 + \s_2 \frac {\abs x}{\abs z} - \frac {\abs {x-z}}{\abs z} \right) \abs z \abs \x\\
& \geq \left( \s_2 - \s_1  - (1+\s_2)\frac {\abs {x-z}}{R_1} \right) \abs z \nu_0\\
& \geq \abs z \frac {\n_0(\s_2-\s_1)}2 .
\end{align*}
\end{proof}

Now we can estimate the solution $u_h$ in an incoming region. We recall the following result:

\begin{proposition} \label{prop-incoming}
Let $\tilde R > 0$, $0 \leq \tilde \nu < \nu $ and $- 1 <  \tilde \s < \s < 1$. Then there exists $R > \tilde R$ such that for $\o \in \symb_0(\R^{2n})$ supported outside $\zoneS_- (\tilde R,\tilde \nu,- \tilde \s)$ and $\o_- \in \symb_0 (\R^{2n})$ supported in $\zoneS_-(R,\nu,-\s)$, we have
\[
\sup_{z \in  {\C_{J,+}}} \nr{\pppg x ^{-\d} \Op(\o_-) (\hh-z)\inv  \Op (\o) \pppg x ^{-\d}}_{\Lc(L^2(\R^n))} = \bigo h 0 (h^\infty).
\]
Moreover the estimate remains true for the limit ${(\hh-(\l +i0))\inv}$, $\l \in J$.
\end{proposition}

This theorem is proved in \cite{robertt89} for the self-adjoint case and extended in \cite{art-mesure,art-nondiss} for our non-selfadjoint setting. Before giving an estimate of $u_h$ in an incoming region, we recall that it concentrates on the hypersurface of energy $E_0$.
For $h\in ]0,h_0]$, $t\geq 0$ and $z \in \C$ we set 
\[
\uh (t,z) = e^{-\frac {it}h(\hh-z)}.
\]

\begin{proposition} \label{prop-loc-E0} 
Let $q \in \symbor(\R^{2n})$ be a symbol which vanishes on $p\inv(I)$ for some neighborhood $I \subset J$ of $E_0$. Let $T \geq 0$. Then there exists $C \geq 0$ such that for $h \in ]0,1]$, $z \in \C_{I,+}$ and $\h \in C_0^\infty(\R,[0,1])$ non-increasing and supported in $[0,T+1]$ we have 
\[
\nr{\frac ih \int_0^\infty \h(t)  \Opw(q) \uh(t,z) \, dt }_{\Lc(L^2(\R^n))} \leq C
\]
and
\[
 \nr{ \Opw(q) (\hh-z)\inv } _{\Lc(L^{2,\d}(\R^n),L^{2,-\d}(\R^n))} \leq C . 
\]
In particular if $q \in C_0^\infty(\R^{2n})$ is supported outside $p \inv(\singl {E_0})$ we have 
\[
 \innp{\Opw(q) u_h}{u_h} \limt h 0 0.
\]

\end{proposition}

This is Proposition 2.11 in \cite{art-mesure}. Note that the same holds if $\Opw(q)$ is on the right of the propagator or the resolvent.

\begin{proposition} \label{prop-mes-zoneS}
Let $\s \in ]\s_1,1[$. Then there exists $R \geq 0$ such that for $q \in C_0^\infty(\R^{2n})$ supported in $\zoneS_-(R,0,-\s)$ we have
\[
\innp{\Opw(q) u_h}{u_h} \limt h 0 0.
\]
\end{proposition}

\begin{proof}
Let $\s_2,\s_3$ be such that $\s_1 < \s_2 < \s_3 < \s$, $R_2 > R_1$ and $\n_0 \in ]0,(\inf J)/3[$. Let $\o_- \in \symb_0(\R^{2n})$ be supported in $\zoneS_-(R_2,\n_0,-\s_2)$ and equal to 1 on $\zoneS_-(2 R_2,2 \n_0, - \s_3)$. Let $R > 2 R_2$ be chosen large enough and consider $q_- , \tilde q_- \in C_0^\infty(\R^{2n})$ supported in $\zoneS_-(R,0,-\s)\cap p\inv(J)$ and such that $\tilde q_- = 1$ on a neighborhood of $\supp q_-$. If $R$ is large enough we have $\zoneS_-(R,0,-\s)\cap p\inv(J) \subset \zoneS_-(R , 3 \n_0 ,-\s)$, so according to Propositions \ref{prop-microloc++} and \ref{prop-incoming} we have
\begin{eqnarray*}
\lefteqn{\nr{\Opw(q_-) (\hh-(E_h+i0))\inv S_h}_{L^2(\R^n)}}\\
&& \leq \nr{\Opw(q_-) (\hh-(E_h+i0))\inv}_{\Lc(L^{2,\d}(\R^n),L^2(\R^n))} \nr{ \Op(\o_-) S_h}_{L^{2,\d}(\R^n)}\\
&& \quad + \nr{\Opw(q_-) (\hh-(E_h+i0))\inv (1-\Op(\o_-)) \pppg x ^{-\d}}_{\Lc(L^2(\R^n))} \nr{ S_h}_{L^{2,\d}(\R^n)}\\
&& = \bigo h 0 \big( \sqrt h \big).
\end{eqnarray*}
The same applies to $\tilde q_-$ and this finally gives
\[
\abs{\innp{\Opw(q_-)u_h}{u_h}} 
\leq  \nr{\Opw(q_-)u_h} \nr {\Opw(\tilde q_-) u_h} + \bigo h 0 (h^\infty)  \limt h 0 0. \qedhere
\]
\end{proof}

\section{Control of large times and of the source term far from the origin}   \label{sec-approx}

As mentionned in introduction, we expect that the semiclassical measure of $(u_h)_{h\in ]0,h_0]}$ on some bounded subset of $\R^{2n}$ does not depend on the values of the amplitude $A(z)$ for large $\abs z$. When restricting our attention to finite times, this is a consequence of Egorov's Theorem and the following proposition, proved in \cite{art-nondiss} (Proposition 2.1):

\begin{proposition} \label{prop-non-entr}
Let $E_2 \geq E_1 > 0$, $J \subset [E_1,E_2]$ and $\s_3 \in [0,1[$ such that $\s_3^2 E_2 < E_1$. Then there exist $\Rc > 0$ and $c_0>0$ such that
\[
 \forall t \geq 0, \forall (x,\x) \in \zoneS_\pm (\Rc,0,\mp \s_3) \cap p\inv(J), \quad \abs{X(\pm t,x,\x)} \geq c_0 (t + \abs x).
\]
\end{proposition}

For $r > 0$ we set
\[
 B_x(r) = \singl{(x,\x)  \st \abs x < r} \subset \R^{2n}.
\]

With Proposition \ref{prop-non-entr} we can check that on a bounded subset of $\R^{2n}$ we can ignore the contribution of the source far from the origin:

\begin{proposition} \label{prop-source-loin}
Let $r > 0$. There exists $R_0 > 0$ such that for $T\geq 0$, $R \geq R_0$, $\Im z \geq 0$ and $q \in C_0^\infty(\R^{2n},[0,1])$ supported in $B_x(r)$ we have
\[
 \nr{\frac ih \int_0^\infty \h_{T,R}(t)\Opw(q) \uh(t,z) ( S_h - \SR )\, dt }_{L^2(\R^n)} = \bigo h 0 (\sqrt h),
\]
where the size of the rest depends on $T$ but not on $z$, $q$ or $R \geq R_0$.
\end{proposition}

We recall that $\h_{T,R} \in C^\infty(\R, [0,1])$ is non-increasing, equal to 1 on $]-\infty, T]$ and equal to 0 on $[T+1,+\infty[$.

\begin{proof}
Let $\Rc$ and $c_0$ be given by Proposition \ref{prop-non-entr} applied with $\s_3 = (1+ \s_1)/2$ (which is allowed according to Assumption \eqref{hyp-J}). Let $\tilde \Rc > \max (\Rc , 2r/c_0, R_1)$ be so large that
\[
\n_0^2 := \inf J - \sup_{\abs x \geq \tilde \Rc}\abs{V_1(x)} >0.
\]
Let $\s_2 \in ]\s_1,\s_3[$.
We consider $\o \in \symbor(\R^{2n})$ supported in $\zoneS_+ (\tilde \Rc , 0, -\s_3)$ and equal to 1 in a neighborhood of $\zoneS_+ (2 \tilde \Rc , 0 , -\s_2)$. Given $\th \in C_0^\infty(\R)$ supported in $J$ and equal to 1 on a neighborhood of $E_0$ we prove that
\[
\sum_{j=1}^3 \nr{\frac ih \int_0^\infty \h_{T,R}(t)\Opw(q) \uh(t,z) B_j(h) ( S_h - \SR )\, dt }_{L^2(\R^n)} = \bigo h 0 \big( \sqrt h \big ) , 
\]
where
\[
B_1(h) =  \Opw(1-\th\circ p), \quad B_2(h) = \Opw(\th\circ p) \Opw(\o) ,\quad B_3(h) = \Opw(\th\circ p) \Opw(1-\o).
\]
The first term is estimated with Propositions \ref{prop-loc-E0} and \ref{prop-norme-S}. The function $\h_{T,R}$ depends on $R$, but since it is non-increasing and always vanishes on $[T+1,+\infty[$, we can check that the estimate is actually uniform in $R > 0$.
According to Proposition \ref{prop-non-entr} and Egorov's Theorem (see Theorem 7.2 in \cite{art-nondiss}), the second term is of size $O(h^\infty)$, again uniformly in $R \geq 0$. 
For the third term we set $R_0 = 3 \tilde \Rc$ and choose $R \geq R_0$. The symbol $(\th \circ p) (1- \o)$ is supported in 
\[
B_x(2 \tilde \Rc) \cup \zoneS_-(2 \tilde R , 0, -\s_2) \cap p\inv(J) \subset B_x(2 \tilde \Rc) \cup \zoneS_-(2 \tilde R , \n_0, -\s_2).
\]
Since $A- A_R$ is supported outside $B(0,3\tilde \Rc)$ and according to Propositions \ref{prop-microloc++} and \ref{prop-microloc-inf} (applied with $A-A_R$ instead of $A$) we obtain that $\nr{\Op(1-\o) (S_h-S_h^R)}_{L^2(\R^n)} = O\big( h^{\frac 32} \big)$ uniformly in $R$.
\end{proof}

We know that \eqref{lim-opquu} holds for some measure $\m_T^R$ when $u_h$ is replaced by $u_h^{T,R}$ (see Theorem 4.3 in \cite{art-mesure}). We now check that $u_h^{T,R}$ is actually a good approximation of $u_h$ (in some sense) when $h > 0$ is small and $T,R$ are large enough:

\begin{proposition}  \label{tps-grands}
Let $r > 0$ and $R_0 > 0$ given by Proposition \ref{prop-source-loin}.
Let $K$ be a compact subset of $B_x(r) \cap p\inv(J)$ and $\e > 0$.
Then there exists $T_0 \geq 0$ such that for $q \in C_0^\infty(\R^{2n})$ supported in $K$, $T \geq T_0$ and $R \geq R_0$ we have
\[
\limsup_{h\to 0}  \abs{\innp{\Opw(q)u_h}{u_h} - \innp{\Opw(q) \uhtr}{\uhtr}} \leq \e \nr q _\infty.
\]
\end{proposition}

This proposition relies on the following consequence of the non-selfadjoint version of Egorov's Theorem (see \cite[Prop.7.3]{art-nondiss}):

\begin{proposition}  \label{prop-super-egorov}
Let $J$ be a neighborhood of $E$ such that Assumption \eqref{hyp-amfaible} holds for all $\l \in J$. Let $K_1$ and $K_2$ be compact subsets of $p\inv (J)$. Let $\e > 0$. Then there exists $T_0 \geq 0$ such that for $T \geq T_0$ and $q_1,q_2 \in C_0^\infty(\R^{2n})$ respectively supported in $K_1$ and $K_2$ we have
\[
\limsup_{h \to 0} \nr{\Opw(q_1) U_h(T) \Opw(q_2)}_{\Lc(L^2(\R^n))} \leq \e \nr {q_1}_\infty \nr {q_2}_\infty.
\]
\end{proposition}

We also need the following result about the classical flow:

\begin{proposition} \label{prop-BzoneS}
 Let $E_1,E_2 \in \R_+^*$ be such that $E_1 \leq E_2$, and $\s \in [0,1[$. If $\Rc$ is chosen large enough then for any compact subset $K$ of $p\inv([E_1,E_2])$ there exists $T_0 \geq 0$ such that 
\[
 \forall w \in K, \forall t \geq T_0, \quad \vf^{\pm t}(w) \in B_x(\Rc) \cup \zoneS_ \pm (\Rc,0,\pm \s).
\]
\end{proposition}

This is slightly more general than Lemma 5.2 in \cite{art-mesure}. We recall the idea of the proof:

\begin{proof}
We consider $\Rc_0$ such that
\[
\forall x \in \R^n, \quad \abs x \geq \Rc_0 \quad \impl \quad \abs{V_1(x)} + \abs x \abs {\nabla V_1(x)} \leq \frac {E_1}3 (1-\s^2).
\]
Let $\t > 0$ be such that $\int_0^\t \frac {(1-\s^2) \sqrt {E_1}}{\sqrt 3 ( \Rc_0 + 4s \sqrt {E_2})} \, ds > \s$. We set $\Rc = \Rc_0 + 4\t \sqrt {E_2}$ and $\Uc_\pm = B_x(\Rc_0) \cup \zoneS_\pm(\Rc , 0 , \pm \s)$.
We first prove that for any $w \in K$, if $\vf^{\pm t}(w) \in \Uc_\pm$ for some $t_w \geq 0$, then $\vf^{\pm t}(w) \in  B_x(\Rc) \cup \zoneS_ \pm (\Rc,0,\pm \s)$ for all $t \geq t_w$. Since $\vf^{\pm (t-t_w)}$ maps $\zoneS_\pm(\Rc , 0 , \pm \s)$ into itself for all $t \geq t_w$, we can assume that $\abs{X(\pm t_w,w)} = \Rc_0$, $\abs {X(\pm t,w)} = \Rc$ and $\abs{X(\pm s,w)} \in [\Rc,\Rc_0]$ for $s \in [t_w,t]$. Assume by contradiction that $\vf^{\pm s}(w) \notin \zoneS_ \pm (\Rc,0,\pm \s)$ when $s \in [t_w,t]$. Then we can check that 
\[
 \pm \frac {\partial}{\partial s} \frac {X(\pm s,w ) \cdot \Xi (\pm s,w )}{\abs{X(\pm s,w )} \abs {\Xi(\pm s,w)}} \geq \frac {(1-\s^2) \sqrt {E_1}} {\sqrt 3 (\Rc_0 + 4(s-t_0)\sqrt {E_2})},
\]
which gives a contradiction. Then it only remains to check that if $T_0$ is chosen large enough, then for all $w \in K$ we can find $t_w  \in [0,T_0]$ such that $\vf^{\pm t_w}(w) \in \Uc_\pm$. For this we use compactness of $K$ and the fact that any trajectory has a limit point in $\O_b([E_1,E_2]) \subset B_x(\Rc_0)$ or goes to infinity and meets $ \zoneS_\pm(\Rc , 0 , \pm \s)$ when $t$ is large enough.
\end{proof}

Let $\s_2 < \s_3 < \s_4 < \s_5 \in ]\s_1,1[$ and $\nu_0 \in ]0, \sqrt{\inf J}/4[$. Let $\Rc$ be given by Proposition \ref{prop-incoming} applied with $(\tilde R , \tilde \nu ,\tilde \s) = (3R_1 , 2\n_0 ,\s_3)$ and $(\nu , \s) = (3 \nu_0 , \s_4)$ . Choosing $\Rc$ larger if necessary, we can assume that $\abs \x \geq 4 \nu_0$ if $p (x,\x) \in J$ and $\abs x \geq \Rc$. We can also assume that  $2\Rc$ satisfies the conclusion of Proposition \ref{prop-BzoneS} applied with $[E_1,E_2] \supset J$.

\begin{lemma} \label{lem-tps-gd}
Let $r \geq 4 \Rc$ and $R_0 > 0$ given by Proposition \ref{prop-source-loin}. Let $Q \in C_0^\infty(\R^{2n},[0,1])$ be supported in $B_x(r) \cap p\inv(J)$ and equal to 1 in a neighborhood of $  B_x(3\Rc)\cap p\inv(\bar I)$ for some open neighborhood $I$ of $E_0$. Let $K$ be a compact subset of $B_x(r) \cap p\inv (J)$ and $\d > 0$.
Then there exists $T_0 \geq 0$ such that for $T \geq T_0$, $R \geq R_0$ and $q \in C_0^\infty(\R^{2n})$ supported in $K$ we have
\[
\limsup_{h\to0} \nr{\Opw(q) \left(u_h -  \uhtr - A_T^\d(h) \Opw (Q) u_h \right) }_{L^2(\R^n)} \leq \d \nr q _\infty,
\]
where $A_T^\d(h)$ is a bounded operator such that
\[
\forall T \geq T_0,\quad \limsup_{h\to 0} \nr{A_T^\d(h)}_{\Lc(L^2(\R^n))} \leq \d.
\]
\end{lemma}

For the proof we follows the same general idea as in \cite{art-mesure}:

\begin{proof} \setcounter{stepproof} 0
\stepp
We consider $\tilde q \in C_0^\infty(\R^{2n})$ supported in $B_x(r)\cap p\inv(J)$ and equal to 1 on a neighborhood of $K$. Let $\th \in C_0^\infty(\R^n)$ be supported in $B(0,3\Rc)$ and equal to 1 on $B(0,2\Rc)$. Let $\o_- \in \Sc_0 (\R^{2n})$ be equal to 1 on $\zoneS_-\left(2\Rc,4 \nu_0 ,-\s_5\right)$ and supported in $\zoneS_-(\Rc,3\nu_0 ,- \s_4)$.
Let $q \in C_0^\infty(\R^{2n})$ be supported in $K$, $h\in ]0,h_0]$, $T \geq 0$ and $z \in\C_{I,+}$ ($h_0 > 0$ was fixed small enough in the introduction). Since $z$ is not in the spectrum of $\hh$ we can consider $(\hh-z)\inv S_h \in H^2(\R^n)$ and write:
\begin{eqnarray} 
\lefteqn{\Opw(q) (\hh-z)\inv S_h - \Opw(q) \uh(T,z)(\hh-z)\inv S_h} \label{dec-tps-gds}\\
\nonumber && = \frac ih \int_0^\infty \Opw(q) \big(\h_{T,R}(t) - \h_{0,R}(t) \uh(T,z) \big) \uh(t,z) S_h \, dt \\
\nonumber && = \frac ih \int_0^\infty  \Opw(q) \big(\h_{T,R}(t) - \h_{0,R}(t) \uh(T,z) \big)  \uh(t,z) S_h^R \, dt  + \bigo h 0 (\sqrt h),
\end{eqnarray}
where the rest is estimated in $L^2(\R^n)$ uniformly in $R \geq R_0$ but not in $T$ (see Proposition \ref{prop-source-loin}).\\

\stepp
We have
\begin{eqnarray*} 
\lefteqn{\Opw(q) \uh (T,z) (\hh -z)\inv S_h=  \Opw(q) \uh (T,z) \Opw(Q) (\hh -z)\inv S_h}\\
&\hspace{2cm} & \quad + \Opw(q) \uh (T,z)\Opw(1-Q) \th(x)   (\hh -z)\inv S_h\\
&& \quad + \Opw(q) \uh (T,z) \Opw(1-Q) (1-\th(x)) \Op (\o_-) (\hh -z)\inv S_h \\
&& \quad + \Opw(q) \uh (T,z) \Opw(1-Q) (1-\th(x)) \Op (1-\o_-) (\hh -z)\inv S_h.
\end{eqnarray*}
The second term of the right-hand side is of size $O(\sqrt h)$ uniformly in $z \in \C_{I,+}$ according to Proposition \ref{prop-loc-E0}. Let $\o \in \symb_0(\R^{2n})$ be supported in $\zoneS_-(2 R_1,\n_0,-\s_2)$ and equal to 1 in $\zoneS_-(3R_1, 2\n_0 , -\s_3)$. According to Propositions \ref{prop-microloc++}, \ref{prop-microloc-inf} and \ref{prop-incoming} we have
\begin{eqnarray*}
\lefteqn{\nr{\Opw(q) \uh(T,z) \Opw(1-Q) (1-\th(x)) \Op (\o_-) (\hh-z)\inv S_h}}\\
&& \leq c_{q,T} \nr{\Op (\o_-) (\hh-z)\inv (1-\Op(\o)) S_h }_{L^{2,-\d}(\R^n)} + \bigo h 0 \big( \sqrt h \big)\\
&& = \bigo h 0 \big( \sqrt h \big) ,
\end{eqnarray*}
uniformly in $z \in \C_{J,+}$ ($\Vert \Opw(q) \uh(T,z) \pppg x ^\d \Vert= O(1)$ uniformly in $z \in \C_{J,+}$) but not in $T$. Finally the last term is of size $O(h^\infty)$ according to Egorov's Theorem and Proposition \ref{prop-BzoneS} if $T \geq T_0$, $T_0$ being given by Proposition \ref{prop-BzoneS} applied with $\s_5$.
We consider $\tilde q, \tilde Q \in C_0^\infty(\R^{2n},[0,1])$ supported in $p\inv(J)$ and equal to 1 respectively on a neighborhood of $\supp q$ and $\supp Q$, and we set
\[
A_T^\d(z,h) = \Opw(\tilde q) \uh(T,z) \Opw(\tilde Q).
\]
According to Proposition \ref{prop-super-egorov} we have
\[
\limsup_{h\to 0} \sup_{z\in \C_{I,+}} \nr{A_T^\d(z,h)}_{\Lc(L^2(\R^n))} \leq \d
\]
when $T \geq T_0$, if $T_0$ was chosen large enough. We finally obtain
\[
\Opw(q) \uh (T,z) (\hh -z)\inv S_h =  \Opw(q) A_T^\d(z,h) \Opw(Q) (\hh-z)\inv S_h + \bigo h 0 (\sqrt h)
\]
in $L^2(\R^n)$ where the size of the rest is uniform in $z \in \C_{I,+}$.\\

\noindent
{\bf 3.}
For $h \in ]0,h_0]$ and $T \geq T_0$, we can take the limit $z \to E_h$ in \eqref{dec-tps-gds} ($E_h \in \bar {\C_{I,+}}$ if $h_0$ is small enough). This gives in $L^2(\R^n)$:
\[
\Opw(q) u_h = \Opw(q) \uhtr - \Opw(q) \uhe(T) \uhor + \Opw(q) A_T^\d(E_h,h) \Opw(Q) u_h + \bigo h 0 (\sqrt h).
\]
Let $q_1 \in C_0^\infty(\R^{2n},[0,1])$ be equal to 1 on a neighborhood of $\singl{\vf^t(z,\x) , t \in [0,1], (z,\x) \in N_E^{R_0}\G}$.
Using the results about the contribution of small times (see in particular Corollary 4.4 in \cite{art-mesure}), we know that
\begin{align*}
 \limsup_{h \to 0} \nr{\Opw(q) \uhe(T) u_h^{0,R_0}}_{L^2(\R^n)}
& =  \limsup_{h \to 0} \nr{\Opw(q) \uhe(T) \Opw(q_1^2) u_h^{0,R_0}}_{L^2(\R^n)} \\
& \leq C \limsup_{h \to 0} \nr{\Opw(q_1) \uhe(T) \Opw(q_1)}_{\Lc(L^2(\R^n))},
\end{align*}
for some constant $C$ which does not depend on $h$, $T$ or $R$ . According to Proposition \ref{prop-super-egorov}, this limit is less than $\d \nr q _\infty$ for any $T \geq T_0$ if $T_0$ was chosen large enough. Since for any $T \geq T_0$ we have
\[
 \vf^{-T}(\supp q) \cap \singl{\vf^t(z,\x) , t \in [0,1], (z,\x) \in N_E^{R}\G \setminus N_E^{R_0}\G} = \emptyset
\]
Egorov's Theorem also gives
\[
  \limsup_{h \to 0} \nr{\Opw(q) \uhe(T) (\uhor - u_h^{0,R_0})}_{L^2(\R^n)} =0,
\]
which concludes the proof.
%
%
%
%
%
%
\end{proof}

Then Proposition \ref{tps-grands} is proved exactly as in \cite{art-mesure} (see Proposition 5.5), and we can show existence of a semiclassical measure:


\begin{proposition} \label{prop-lim-mes}
There exists a non-negative Radon measure $\m$ on $\R^{2n}$ such that for $q \in C_0^\infty(\R^{2n})$ we have
\begin{equation} \label{double-lim}
\int_{\R^{2n}} q \, d\m =  \lim_{T,R \to +\infty}  \int_{\R^{2n}} q \, d\m_T^R
\end{equation}
and
\[
\innp{\Opw(q) u_h}{u_h} \limt h 0 \int_{\R^{2n}} q \, d\m.
\]
\end{proposition}

%

\begin{proof}
The result is clear outside $p\inv(\singl {E_0})$, so we focus on symbols supported in $p\inv(J)$. Let $K$ be a compact subset of $p\inv(J)$ and $\e > 0$. Let $T_0$ and $R_0$ given by Proposition \ref{tps-grands}. For $T_1,T_2 \geq T_0$, $R_1,R_2 \geq R_0$ and $q \in C_0^\infty(\R^{2n})$ supported in $K$ we have
\begin{eqnarray*}
\lefteqn{\abs{\int_{\R^{2n}} q \, d\m_{T_1}^{ R_1} -\int_{\R^{2n}} q \, d\m_{T_2}^{R_2}} }\\
&&= \lim_{h \to 0} \abs{\innp{\Opw(q) u_h^{T_1, R_1}} {u_h^{T_1, R_1}} -  \innp{\Opw(q) u_h^{T_2, R_2}} {u_h^{T_2, R_2}} } \leq 2 \e \nr q _\infty.
\end{eqnarray*}
This proves that $(T,R) \mapsto  \int q \,d\m_T^{ R}$ has a limit at infinity, which we denote by $L(q)$. The map $q \mapsto \int q \,d\m_T^{R}$ is a nonnegative linear form on $C_0^\infty(\R^{2n})$ for all $T,R > 0$, and hence so is $q \mapsto L(q)$. Let $T_0$ be as above for $\e = 1$ and $C_K$ be a constant such that for all $q \in C_0^\infty(\R^{2n})$ we have
\[
\abs{\int_{\R^{2n}} q \, d\m_{T_0}^{ R_0}} \leq C_K \nr q _\infty.
\]
Then we have
\begin{align*}
\abs{L(q)}
& \leq  \abs{ L(q) - \int_{\R^{2n}} q \, d\m_{T_0}^{R_0}} + \abs{\int_{\R^{2n}} q \, d\m_{T_0}^{R_0}} \\
& \leq  \lim_{T,R \to +\infty} \abs{ \int_{\R^{2n}} q \, d\m_{T}^{R} - \int_{\R^{2n}} q \, d\m_{T_0}^{R_0}} + C_K \nr q _\infty \\
& \leq (2+C_K) \nr q_\infty,
\end{align*}
which proves that the linear form $L$ on $C_0^\infty(\R^{2n})$ can be extended as a continuous linear form on the space of continuous and compactly supported functions on $\R^{2n}$. The first assertion is now a consequence of Riesz' Lemma. And the second can be proved as in \cite{art-mesure}, using the fact that $\innp{\Opw(q) u_h} {u_h}$ is close to $\innp{\Opw(q) u_h^{T,R}} {u_h^{T,R}}$ in the sense of Proposition \ref{tps-grands}, and $\innp{\Opw(q) u_h^{T,R}} {u_h^{T,R}}$ goes to $\int q \, d\m_{T}^R$ as $h$ goes to 0.
\end{proof}

It is now easy to prove the remark about the measures $\tilde \m _R$ mentionned in introduction:

\begin{proposition} \label{prop-lim-mes++}
Let $r > 0$, $R_0$ given by Proposition \ref{prop-source-loin} and $R \geq R_0$. Then the measures $\m$ and $\tilde \m _R$ coincide on $B_x(r)$.
\end{proposition}


To prove this assertion, we only have to apply Proposition \ref{tps-grands} with $u_h$ and $\tilde u_h^R$. Let $q \in C_0^\infty(\R^{2n})$ be supported in $B_x(r) \cap p\inv(J)$. Since for large $T$ and small $h$ the quantity $\innp{\Opw(q) u_h^{T,R}}{u_h^{T,R}}$ is a good approximation both for $\innp{\Opw(q) u_h}{u_h}$ and $\innp{\Opw(q) \tilde u_h^{R}}{\tilde u_h^{R}}$, these two quantities have the same limit as $h$ goes to 0. \\

It only remains to prove the properties given in Theorem \ref{th-mesure-nb}:

%
%
%
%
%
%

%
%

\begin{proof}
The first two properties are direct consequences of Propositions \ref{prop-loc-E0} and \ref{prop-mes-zoneS}.
Let $r > 0$ and $R_0$ given by Proposition \ref{prop-source-loin}. Property (c) is already known for the compactly supported amplitude $A_{R_0}$, and hence for any $q \in C_0^\infty(\R^{2n})$ we have
\[
\int_{\R^{2n}} (-H_p + 2 \Im \tilde E + 2 V_2) q \, d\tilde \m_{R_0} = \pi (2\pi)^{d-n}\int_{\negg} q(z,\x)   \abs{A_{R_0}(z)}^2 \abs \x \inv \abs{\hat S (\x)}^2 \, d\snegg(z,\x).
\]
Now assume that $q$, and hence $(H_p + 2 \Im \tilde E + 2V_2)q$, are supported in $B_x(r)$. According to Proposition \ref{prop-lim-mes++} we have
\[ 
\int_{\R^{2n}} (-H_p + 2 \Im \tilde E + 2 V_2) q \, d\tilde \m_{R_0} = \int_{\R^{2n}} (-H_p + 2 \Im \tilde E + 2 V_2) q \, d\m.
\]
On the other hand, according to Proposition \ref{prop-non-entr} trajectories comming from the points of $\negg \setminus N_E^{R_0}\G$ never reach $B_x(r)$, so we also have
\[
 \pi (2\pi)^{d-n}\int_{\negg} q(z,\x)  (\abs{A(z)}^2 - \abs{A_{R_0}(z)}^2) \abs \x \inv \abs{\hat S (\x)}^2 \, d\snegg(z,\x) = 0.
\]
This proves that Property (c) holds when $q \in C_0^\infty(\R^{2n})$ is supported in $B_x(r)$. Since this holds for any $r > 0$, the theorem is proved.
\end{proof}

\appendix

\section{Short review about differential geometry} \label{sec-II}

We briefly recall in this appendix the basic results of differential geometry we have used. Detailed expositions can be found for instance in \cite{docarmo} and \cite{spivak4}.\\

Let $M$ be a differential manifold. We denote by $\vfield(M)$ the set of vector fields on $M$. An affine connection on $M$ is a mapping
\[
 \nabla : \left\{ \begin{array}{ccc} \vfield (M) \times \vfield (M) &\to& \vfield(M) \\ (X,Y) & \mapsto & \nabla_X Y \end{array} \right.
\]
which satisfies the following properties (for $X,Y,Z \in \vfield (M)$ and $f,g \in C^\infty (M)$):
\begin{enumerate}[(i)]
 \item $\nabla _{fX+gY}Z = f \nabla _X Z + g \nabla _Y Z$,
\item $\nabla _X(Y+Z) = \nabla _X Y + \nabla _X Z$,
\item $\nabla _X (fY) = f \nabla _X Y + (X\cdot f) Y$. \\
\end{enumerate}

The Levi-Civita connection on $M$ is the unique connection $\nabla$ on $M$ which is
\begin{enumerate}[(i)]
 \item symmetric:
\[
 \forall X,Y \in \vfield (M) , \quad \nabla _X Y - \nabla _Y X = [X,Y] = XY-YX,
\]
\item and compatible with the Riemannian metric:
\begin{equation} \label{der-innp}
\forall X,Y,Z \in \vfield(M), \quad X \cdot \innp{Y}{Z} = \innp{\nabla _X Y } Z_M + \innp Y {\nabla _X Z}_M.
\end{equation}
\end{enumerate}

The Levi-Civita connection on $\R^n$ endowed with the canonical metric is the usual differential.
%
Now let $\G$ be a submanifold of $\R^n$, endowed by the Riemannian structure defined by restriction of the scalar product of $\R^n$. For $X,Y \in \vfield (\G)$, $z \in \G$, and $\bar X ,\bar Y \in \vfield (\R^n)$ such that $X= \bar X$ and $Y = \bar Y$ in a neighborhood of $z$ in $\G$ we set
\[
 \nabla_X ^\G Y (z)= \big ( \nabla_{\bar X}^{\R^n} \bar Y (z)\big) _z^T.
\]
This definition does not depend on the choice of $\bar X $ or $\bar Y$ and defines the Levi-Civita connexion on $\G$.\\

Let $X \in \Xc(\G)$. The divergence $\divg X (z)$ at point $z \in \G$ is defined as the trace of the linear map $Y \mapsto \nabla_Y^\G X (z)$ on $T_z \G$. If $X \in \Xc(\G)$ and $f \in C^\infty(\G)$ then we have
\begin{equation} \label{div-fX}
 \divg(fX) = X \cdot f + f \divg X.
\end{equation}

The main theorem we have used in Section \ref{sec-terme-source} is the following:

\begin{theorem}[Green's Theorem] \label{th-green} 
If $M$ is an oriented Riemannian manifold and $X \in \Xc(M)$ is compactly supported, then 
\[
\int_M \divg X \, dV_M = 0,
\]
where $dV_m$ denotes the volume element on $M$.
\end{theorem}

We finally recall the basic properties of the second fundamental form on $\G$. Given $X,Y \in \Xc(\G)$ and $z \in\G$ we set
\[
 \II _z (X,Y) = \nabla ^{\R_n} _{\bar X} \bar Y(z) - \nabla ^{\G}_ {X}  Y(z) = \left( \nabla ^{\R_n} _{\bar X} \bar Y(z) \right) _z^N \in N_z \G,
\]
where $\bar X$ and $\bar Y$ are extensions of $X$ and $Y$ on a neighborhood of $z$ in $\R^n$. We can check that $\II_z(X,Y)$ is well-defined and actually only depends on $X(z)$ and $Y(z)$. Moreover the bilinear form $\II_z$ is symmetric.

\subsection*{Acknowledgement} This work is partially supported by the French National Research Project NOSEVOL, No. ANR 2011 BS01019 01, entitled {\em Nonselfadjoint operators, semiclassical analysis and evolution equations}.

\bibliographystyle{alpha}
\bibliography{bibliotex}

\vspace{1cm}

\begin{center}

\begin{minipage}{0.8 \linewidth}

\noindent  {\sc Institut de mathématiques de Toulouse}

\noindent  {\sc 118, route de Narbonne}

\noindent  {\sc 31062 Toulouse Cédex 9}

\noindent  {\sc France}

 \begin{verbatim}
julien.royer@math.univ-toulouse.fr
\end{verbatim}
\end{minipage}
\end{center}

\end{document}